\documentclass[a4paper, 11pt, oneside]{amsart}
\pdfoutput=1
\usepackage[utf8]{inputenc}
\usepackage[T1]{fontenc}
\usepackage[english]{babel}
\usepackage{enumitem}
\usepackage{graphicx}
\usepackage{caption, subcaption}
\usepackage{booktabs}
\usepackage{amsfonts}
\usepackage{mathtools}
\usepackage[thinc]{esdiff}
\usepackage[a4paper,top=2.5cm,bottom=2.5cm,left=2.5cm,right=2.5cm]{geometry} 
\usepackage{amsmath, amsthm, amssymb, mathrsfs}
\usepackage[hidelinks]{hyperref}
\usepackage[usenames, dvipsnames]{xcolor}
\definecolor{lime}{RGB}{89,223,0}
\usepackage{empheq}
\definecolor{celestino}{RGB}{102,178,255}
\definecolor{bordeaux}{RGB}{140,0,0}
\hypersetup{
    colorlinks,
    linkcolor={bordeaux},
    citecolor={lime},
    urlcolor={blue!50!black}
}
\usepackage{dirtytalk}
\usepackage[symbol]{footmisc}
\usepackage{csquotes}
\usepackage{dsfont}
\usepackage[makeroom]{cancel}
\usepackage[backend=biber,style=alphabetic,sorting=none]{biblatex}
\newcommand{\R}{\mathbb{R}}
\newcommand{\N}{\mathbb{N}}
\newcommand{\cov}[1]{\langle{#1}\rangle}

\newcommand{\dfr}{\mathrm{\,d}}

\newcommand{\abs}[1]{\left\lvert{#1}\right\rvert}
\newcommand{\norm}[1]{\left\lVert{#1}\right\rVert}

\newcommand{\ttonde}[1]{({#1})}
\newcommand{\quadre}[1]{\left[{#1}\right]}
\newcommand{\graffe}[1]{\left\lbrace{#1}\right\rbrace}

\newcommand*{\crosssymbol}{%
    \text{%
      \raise 1ex\hbox{%
        \rlap{\vrule height.2pt depth.2pt width .75ex}%
        \hbox to .75ex{\hss\vrule height .5ex depth 1ex\hss}%
      }%
    }%
}
\newcommand*{\crossupsidedown}{%
    \text{%
      \raise .5ex\hbox{%
        \rlap{\vrule height.2pt depth.2pt width .75ex}%
        \hbox to .75ex{\hss\vrule height 1ex depth .5ex\hss}%
      }%
    }%
}

\DeclareMathOperator*{\esssup}{ess\,sup}

\theoremstyle{plain}

\theoremstyle{plain}

\theoremstyle{definition}
\newtheorem{defn}{\textbf{Definition}}[section]

\theoremstyle{plain}
\newtheorem{teo}{\textbf{Theorem}}[section]

\newtheorem{lemma}[teo]{\textbf{Lemma}}

\theoremstyle{remark}
\newtheorem{oss}[subsection]{\textit{Remark}}

\usepackage{soul}

\numberwithin{equation}{section}

\title{A partial stochastic equilibrium model and its limiting behaviour}
\author{Alessandro Prosperi}
\date{}

\begin{document}
\begin{abstract}
The existence of a (partial) market equilibrium price is proved in a complete, continuous time finite-agent market setting. The economic agents act as price takers in a fully competitive setting and maximize exponential utility from terminal wealth.
As the number $N$ of economic agents goes to infinity, the BSDE system of $N$ equations characterizing the equilibrium asset price dynamics decouples.
Due to the system's symmetry, the influence of the mean field of the agents, conditionally on the common noise, becomes deterministic.
\end{abstract}
\maketitle
\raggedbottom
\section{Introduction}
We prove the existence of a (partial) market equilibrium price in a complete,
continuous time finite-agent market setting and analyze the behaviour of
the equilibrium price as the number $N$ of economic agents goes to
infinity. The economic agents act as price takers in a fully
competitive setting and maximize exponential utility from terminal
wealth. The $i$-th agent receives a lump sum at a terminal time $T$
which depends on two independent Brownian motions $B$ and $W_i$. During
the time horizon $[0,T]$, the agent $i$ invests in two asset: $S_0$ and
$S_i$. $S_0$ is the asset whose price will be determined in the
equilibrium, while $S_i$ is a an asset that no one else except agent
$i$ can trade. This type of market can be seen as one in which each
agent is investing in a publicly traded asset, in addition to managing
his/her own company. In addition to all of these assets, the agents
have access to an interest-less money market account.

\medskip

Our approach relies on the solvability of a nonlinear quadratic
BSDE-system of $N$
equations, where $N$ is the number of agents, which characterizes the
equilibrium asset price dynamics. We prove that the system
has a unique solution for each $N$ and investigate the behaviour of
the system as the number $N$ of economic agents goes to infinity.
Using the symmetry of the system together with an approach similar to
representative-agent-type aggregation, we show that a ``sufficient statistic'',
namely the market-price-of-risk process $\lambda$, solves an autonomous one-dimensional BSDE, driven by $N+1$ Brownian motions. It is through the
convergence of solutions to this BSDE that we can show that the whole BSDE
system decouples in the $N\to\infty$ limit: the influence of the
mean field of the agents, conditionally on the common noise $B$, becomes
determistic.

\subsection{Notation and convention} 
For $j,d\in\N$, the set of $j\times
d-$matrices is denoted by $\R^{j\times d}$. The Euclidean space $\R^j$
is identified with the set $\R^{j\times 1}$, i.e., vectors in $\R^j$
are columns by default. $\abs{\cdot}$ will denote the Euclidean norm
on either $\R^{j\times d}$ or $\R^j$.
We work on a finite time
horizon $\quadre{0,T}$ with $T>0$, where $\mathbb{F}=\graffe{\mathcal
{F}_t}_{t\in\quadre{0,T}}$ is the augmented filtration generated by a sequence of independent, $1-$dimensional Brownian motions $\graffe{B}\cup\graffe{W_n}_{n\in\N}$.
{The stochastic integrals in what follows are taken with respect to $B$ or $(W_1,\dots, W_N)$ or
both. 
More precisely, given a vector of $d$ independent brownian motions $\tilde{B}$ (i.e. a $d-$dimensional brownian motion), the stochastic integral with respect to $\tilde{B}$ is taken for $\R^
{1\times d}-$valued (row) processes as if $\dfr \tilde{B}$ were a column of
its components, i.e., $\int \sigma\ttonde{t}\dfr\tilde{B}\ttonde{t}$ stands
for $\sum_{i=1}^d\int{\sigma_j\ttonde{t}}\dfr \tilde{B}_j({t})$. }
Similarly, for a process $Z$ with values in $\R^{j\times d}$,
$\int Z\ttonde{t}\dfr B\ttonde{t}$ is an $\R^j$ valued process whose
components are stochastic integrals of the rows $Z^i$ of $Z$ with
respect to $\dfr B\ttonde{t}$.
We denote by
\begin{itemize}
    \item $\mathcal{S}^p$ the set of all $\mathbb{F}-$adapted continuous
     process $Y$ on $\quadre{0,T}$ such that
    \begin{align*}
        \norm{Y}_{\mathcal{S}^p}:=E\quadre{\sup_{0\leq r \leq T}\abs
         {Y\ttonde{r}}^p}^{\frac{1}{p}}<\infty,
    \end{align*}
    \item $\mathcal{S}^\infty$ the set of all $\mathbb
     {F}-$adapted continuous process $Y$ on $\quadre{0,T}$ such that 
    \begin{align*}
    \norm{Y}_{\mathcal{S}^\infty}:=\esssup_{\omega\in\Omega}\sup_
         {0\leq r\leq T}\abs{\ttonde{Y\ttonde{r}}\ttonde{\omega}}<\infty,
    \end{align*}
    \item $H^p$ the set of $\mathbb{F}-$predictable $\R^d-$valued
     processes $Z$ such that
    \begin{align*}
        \norm{Z}_{H^p}:=E\quadre{\Big({\int_0^T\abs{Z(r)}^2\dfr r}\Big)^
         {\frac{p}{2}}}^{\frac{1}{p}}<\infty
    \end{align*}
    \item $\mathcal{H}^p$ the space of all continuous local martingales $M$, such that
    \begin{align*}
        \norm{M}_{\mathcal{H}^p}:=E\quadre{({\cov{M}_T-\cov{M}_0})^
         {\frac{p}{2}}}^{\frac{1}{p}}<\infty.
    \end{align*}
\end{itemize}
The set of all processes of bounded mean oscillation is denoted by BMO
(we refer the reader to \cite{kazamaki2006continuous} for all the
necessary background on the BMO processes) and the set of all
$\Tilde{B}-$integrable processes $\sigma$ such that $\int \sigma \dfr \Tilde{B}$ is in BMO, where $\Tilde{B}$ is a $d-$dimensional brownian motion, is denoted by bmo.
The set of all $\mathbb{F}-$progressively measurable process is denoted
by $\mathcal{P}$. $\mathcal{P}^r$ denotes the set of all $c\in\mathcal
{P}$ with $\int_0^T\abs{c}^r\dfr t<\infty$ a.s. The same notation is
used for scalar, vector or matrix valued processes - the distinction
will always be clear from the context.

\section{The problem}
\subsection{Model primitives.}
The model primitives can be divided into
 two groups. In the first we postulate the form of the dynamics
 of the traded assets, and in the second we describe the
 characteristics of individual agents.
\subsubsection{The traded assets}
Our market consists of $1$ risk-less asset normalized to a constant process
with value $1$, and
$N+1$ risky assets with price dynamics given by:
\begin{align}
    \notag&S_0\ttonde{t} = S_0\ttonde{0} + \int_0^t S_0\ttonde
     {u}\lambda\ttonde{u}\dfr u + \int_0^t S_0\ttonde{u}\dfr B\ttonde
     {u},&&\forall t\in\quadre{0,T},\\
    \label{altririskyasset}&S_n\ttonde{t} = S_n\ttonde
     {0} + \int_0^tS_n\ttonde{u}\dfr W_n\ttonde
     {u},\qquad n=1,\dots,N.&&\forall t\in\quadre{0,T}.
\end{align}
We will simply write $S$ to refer to the vector of all risky assets prices.

\subsubsection{Agents}
There is a finite number $N\in\N$ of
economic agents, and agent $i$ invests in two assets: $S_0$ and $S_i$.
Moreover each economic agent is characterized by the so called \textbf{random-endowment}. Each agent receives a lump sum
$e^i_T = e_T\ttonde{B\ttonde{T}, W_i\ttonde{T}} $ at time $T$, for
some $\alpha$-Hölder continuous bounded function $e_T:\R\times\R\to\R$.
\begin{oss}
We do not distinguish agents based on their initial risk-less asset share holdings, as it's usually done in other equilibrium models.
Indeed, as the reader can easily notice from the following characterization, the optimal portfolio strategy does not depend on the initial holding. 
The reader can just think about the initial holding of agent $i$ as a constant added to his/her {random-endowment} $e^i$.
\end{oss}
Finally each agent wants to maximize the expected utility arising from the exponential utility function
\begin{align}\label{utfunagent}
    U\ttonde{c} = -e^{-c}\qquad\text{for }c\in\R.
\end{align}
\begin{oss}
    Considering a utility function of the form (\ref{utfunagent}) for every agent, we implicitly assume, without loss of generality, that each agent's absolute risk-aversion coefficient is normalized to $1$.
\end{oss} 
\subsection{Admissibility and equilibrium}
\begin{defn}
A progressively measurable process ${\pi}=\ttonde{\pi_0,\dots,\pi_N}$ with values in $\R^{N+1}$ is said to be an \textbf{admissible strategy} if $\pi\in\mathrm{bmo}$. 
The set of all admissible strategies is denoted by $\mathcal{A}$.
\end{defn}

\begin{oss}
Given an admissible strategy $\pi$, the first component ${\pi_0}$ denotes the investment in $S_0$, while for each $i$, $\pi_i$ denotes the investment in the $i-$th asset $S_i$. 
\end{oss}

\begin{oss}
One must check that $\pi_0\lambda\in {\mathcal P}^1$ and $\pi \in {\mathcal P}^2$ in order to have all the integrals well defined. 
The reader should check that this is automatically satisfied when $\pi, \lambda \in \mathrm{bmo}$.
\end{oss}

The \textbf{gains process} is the continuous adapted process $X^{\pi}$ such that $\dfr X^{\pi} := {\pi}\ttonde{t}\cdot\frac{\dfr S\ttonde{t}}{S\ttonde{t}}$ and $X^{\pi}\ttonde{0}=0$.
\begin{defn}
We say that a bmo process $\lambda^*$ is an \textbf{equilibrium market price of risk} if there exist admissible strategies $\hat{\pi}^i\in\mathcal{A}$, $i=1,\dots, N$, such that the following two conditions hold:
\begin{enumerate}
    \item \textit{Single-agent optimality}: For each $i=1,\dots,N$ and all $\pi\in\mathcal{A}$, we have
    \begin{align*}
        E\quadre{U\big({X^{\hat{\pi}^i}\ttonde{T} + e^i\ttonde{T}}\big)}\geq E\quadre{U\ttonde{X^{\pi}\ttonde{T} + e^i\ttonde{T}}}
    \end{align*}
    \item \textit{Partial market clearing}: 
    \begin{align*}
        \sum_{i=1}^N \hat{\pi}_0^i\ttonde{t} = 0,\text{ on }[0,T).
    \end{align*}
\end{enumerate}
\end{defn}
\section{Results}
\subsection{A BSDE characterization} 
Our first result is the characterization of equilibria in terms of a system of $N$ backward stochastic differential equations (BSDE).
A solution to such a system is a pair $\ttonde{Y,Z}$, where the vector process $Y$ and the matrix process $Z$ take values in $\R^N$ and $\R^{N\times(N+1)}$, respectively.
We say that $\ttonde{Y,Z}$ is an $\ttonde{\mathcal{S}^\infty\times\mathrm{bmo}}-$solution if all the components of $Y$ are in $\mathcal{S}^\infty$, and all the components of $Z$ are in bmo.
\begin{teo}\label{BSDEchar}
(A BSDE Characterization)
Suppose that $\ttonde{Y^{\ttonde{N}},Z^{\ttonde{N}}}$ is an ${\mathcal{S}^\infty\times\mathrm{bmo}}-$solution to 
\begin{align}\label{equilibriumNagent}
    \begin{cases}
    \dfr Y^{\ttonde{i,N}}\ttonde{t} &= Z^{\ttonde{i,N}}_1\ttonde
     {t}\dfr B\ttonde{t} +Z^{\ttonde{i,N}}_2\ttonde{t}\dfr W_i\ttonde
     {t} -\big({\frac{1}{2}\ttonde{\lambda^{\ttonde{N}}\ttonde
     {t}}^2 - \lambda^{\ttonde{N}}\ttonde{t}Z^{\ttonde{i,N}}_1\ttonde
     {t}}\big)\dfr t,\\
     Y^{\ttonde{i,N}}\ttonde{T}&=e_T\ttonde{B\ttonde{T},W_i\ttonde{T}}\in L^\infty\ttonde{\Omega, \mathcal{F}_T, \mathbb{P}}.
    \end{cases}
\end{align}
For $t\in [0,T)$ and $i=1,\dots,N.$
Then 
\begin{align*}
    \lambda^{\ttonde{N}}\ttonde{t} = \frac{1}{N}\sum_{i=1}^NZ^{\ttonde{i,N}}_1\ttonde{t}.
\end{align*}
is the equilibrium market price for the $N-$dimensional problem.
\end{teo}
\begin{oss}\label{formavecdelsystNagent}
Defining the matrix $\boldsymbol{z}\in\R^{N\times N+1}$ and the function $\boldsymbol{f}:\R^{N\times N+1}\to\R^N$ by: 
\begin{align*}
    \boldsymbol{z}=\begin{pmatrix}
Z^{\ttonde{1,N}}_1&Z^{\ttonde{1,N}}_2&0&\dots& &0\\
Z^{\ttonde{2,N}}_1&0&Z^{\ttonde{2,N}}_1&0&\dots&0\\
\vdots\\
Z^{\ttonde{N,N}}_1&0&\dots& & 0&Z^{\ttonde{N,N}}_2
\end{pmatrix},
\boldsymbol{f}\ttonde{\boldsymbol{z}}=\begin{pmatrix}
-\frac{1}{2}\ttonde{\boldsymbol{\lambda}\ttonde{\boldsymbol{z}}}^2 + \ttonde{\boldsymbol{\lambda}\ttonde{\boldsymbol{z}}}\boldsymbol{z}_{11}\\
-\frac{1}{2}\ttonde{\boldsymbol{\lambda}\ttonde{\boldsymbol{z}}}^2 + \ttonde{\boldsymbol{\lambda}\ttonde{\boldsymbol{z}}}\boldsymbol{z}_{21}\\
\vdots\\
-\frac{1}{2}\ttonde{\boldsymbol{\lambda}\ttonde{\boldsymbol{z}}}^2 + \ttonde{\boldsymbol{\lambda}\ttonde{\boldsymbol{z}}}\boldsymbol{z}_{N1}\\
\end{pmatrix}
\end{align*}
where $\boldsymbol{\lambda}\ttonde{\boldsymbol{z}}=\frac{1}{N}\sum_{i=1}^N\boldsymbol{z}_{i1}$, we can write the system in (\ref{equilibriumNagent}) as
\begin{align*}
    \dfr\boldsymbol{Y}\ttonde{t} = \boldsymbol{z}\ttonde{t}\dfr \boldsymbol{W}\ttonde{t} + \boldsymbol{f}\ttonde{\boldsymbol{z}\ttonde{t}}\dfr t,\qquad \boldsymbol{Y}\ttonde{T} = \boldsymbol{e}_T.
\end{align*}
Where
\begin{align*}
\ttonde{\boldsymbol{e}_T}^T&:=\ttonde{e_T\ttonde{B\ttonde{T},W_1\ttonde{T}},e_T\ttonde{B\ttonde{T},W_2\ttonde{T}},\dots,e_T\ttonde{B\ttonde{T},W_N\ttonde{T}}},\\
\ttonde{\boldsymbol{Y}}^T &:=\ttonde{ Y^{\ttonde{1,N}},Y^{\ttonde{2,N}},\dots,Y^{\ttonde{N,N}}}\text{ and}\\ \ttonde{\boldsymbol{W}}^T&:=\ttonde{B,W_1,W_2,\dots,W_N}.
\end{align*}
\end{oss}
\begin{proof}
Having fixed an $\mathcal{S}^\infty\times\mathrm{bmo}$ solution $\ttonde{Y,Z}$ we pick an agent $i\in\graffe{1,\dots,N}$ and a strategy $\pi\in\mathcal{A}$ and define the processes $X^i$ and $V^i$ by 
\begin{align*}
    X^i\ttonde{\cdot} = \int_0^\cdot \pi_0\ttonde{s}\frac{\dfr S_0\ttonde{s}}{S_0\ttonde{s}} + \int_0^\cdot \pi_i\ttonde{s}\frac{\dfr S_i\ttonde{s}}{S_i\ttonde{s}},\ V^i\ttonde{\cdot}=-\exp\ttonde{-X^i\ttonde{\cdot}-Y^i\ttonde{\cdot}},
\end{align*}
where $\frac{\dfr S_0\ttonde{t}}{S_0\ttonde{t}}=\lambda^{\ttonde{N}}\ttonde{t}\dfr t + \dfr B\ttonde{t}$ and $S_i\ttonde{t}$ is as in (\ref{altririskyasset}).
The semimartingale decomposition of $V^i$ is given by $\dfr V^i = V^i\mu_V\dfr t + V^i\sigma_V\cdot\ttonde{\dfr B,\dfr W_i}^T $, where 
\begin{align}
    \begin{aligned}\label{dynofV}
    \mu_V &= -\Big( \pi_0\ttonde{t}\lambda^{\ttonde{N}}\ttonde{t} -\frac{1}{2}\big({\ttonde{\pi_0\ttonde{t} + Z^{\ttonde{i,N}}_1\ttonde{t}}^2+\ttonde{\pi_i\ttonde{t} + Z^{\ttonde{i,N}}_2\ttonde{t}}^2}\big) \\
    & \qquad\quad - \big({\frac{1}{2}\ttonde{\lambda^{\ttonde{N}}\ttonde{t}}^2 - \lambda^{\ttonde{N}}\ttonde{t}Z^{\ttonde{i,N}}_1\ttonde{t}}\big) \Big) \\
    \sigma_V &=-(\pi_0\ttonde{t} + Z^{\ttonde{i,N}}_1\ttonde{t}, \pi_i\ttonde{t} + Z^{\ttonde{i,N}}_2\ttonde{t}).
    \end{aligned}
\end{align}
It is straightforward to show that $\mu_V\geq 0$ for all values of $\pi_0$ and $\pi_1$.
Furthermore the coefficients $\mu_V$ and $\sigma_V$ are regular enough to conclude that $V^i$ is a supermartingale for all admissible $\pi$.
Indeed, (\ref{dynofV}) gives us the following
\begin{align*}
V^i\ttonde{t}=\mathcal{E}\Big({\int_0^\cdot\sigma_V\ttonde{r}\cdot\ttonde{\dfr B\ttonde{r},\dfr W_i\ttonde{r}}^T}\Big)_t\exp\Big({\int_0^t\mu_V\ttonde{r}\dfr r}\Big)=:M\ttonde{t}A\ttonde{t}
\end{align*}
with $M\ttonde{t}$ a negative martingale since $\sigma_V\in\mathrm{bmo}$ as seen in \cite[][,Theorem 2.3]{kazamaki2006continuous} and $A\ttonde{t}$ an increasing process.
This allows us to conclude that for each $s<t$ we have 
\begin{align*}
    E_s\quadre{M\ttonde{t}A\ttonde{t}}\leq E_s\quadre{M\ttonde{t}A\ttonde{s}} = A\ttonde{s}E_s\quadre{M\ttonde{t}}=M\ttonde{s}A\ttonde{s}.
\end{align*}
Therefore, 
\begin{align*}
    \sup_{\pi\in\mathcal{A}_{\lambda^{\ttonde{N}}}}E\quadre{U\ttonde{X^i\ttonde{T} + e^i_T}}= \sup_{\pi\in\mathcal{A}_{\lambda^{\ttonde{N}}}}E\quadre{V^i\ttonde{T}}\leq  E\quadre{V^i\ttonde{0}}.
\end{align*}
Next, in order to construct the optimizer, we produce a strategy for which $\mu_V=0$.
More precisely let $\hat{X}^i$ be the unique solution to the following linear SDE:
\begin{align*}
    \hat{X}^i\ttonde{0}= 0,\ \dfr \hat{X}^i\ttonde{t}= ( \lambda^{\ttonde{N}}\ttonde{t} -Z^{\ttonde{i,N}}_1\ttonde{t} )\frac{\dfr S_0\ttonde{t}}{S_0\ttonde{t}} - Z^{\ttonde{i,N}}_2\ttonde{t} \frac{\dfr S_i\ttonde{t}}{S_i\ttonde{t}}
\end{align*}
and set 
\begin{align*}
    \hat{\pi}_0^i\ttonde{t} = \lambda^{\ttonde{N}}\ttonde{t} -Z^{\ttonde{i,N}}_1\ttonde{t},\qquad \hat{\pi}_i^i\ttonde{t}= - Z^{\ttonde{i,N}}_2\ttonde{t}.
\end{align*}
It follows immediately that $\hat{\pi}^i\in\mathcal{A}$ and that the choice of $\hat{\pi}^i$, makes the process $V^i$ a martingale hence the strategy $\hat{\pi}^i$ optimal for agent $i$.
Turning to market clearing, we
compute
\begin{align*}
    \sum_{i=1}^N \hat{\pi}_1^i\ttonde{t} = \sum_{i=1}^N \ttonde{\lambda^{\ttonde{N}}\ttonde{t} -Z^{\ttonde{i,N}}_1\ttonde{t}} = 0, 
\end{align*}
which proves the statement of the theorem.
\end{proof}
\subsection{The representative agent limit equation}
In this section we show that the \emph{representative agent}
\begin{align*}
    \Bar{Y}_N\ttonde{t}={\frac{1}{N}}\sum_{i=1}^NY^{\ttonde{i,N}}
    \ttonde{t}
\end{align*}
process converges as $N\to\infty$ to the solution to the so called
\emph{limiting equation}.
We then get an explicit formula for the equilibrium market price $\lambda$ in the limit.
\begin{teo}\label{limrepagdyn} (Representative agent's limit dynamics)
Suppose $\ttonde{Y,\lambda}$ is an ${\mathcal{S}^\infty\times\mathrm{bmo}}-$solution to the \emph{limiting equation}
\begin{align}\label{repagentdyn}
    \begin{cases}
    \dfr {Y}\ttonde{t} &= \frac{1}{2}\lambda\ttonde{t}^2\dfr t +\lambda\ttonde{t}\dfr B\ttonde{t}, \qquad t\in [0,T)\\
    {Y}\ttonde{T}&=E\quadre{e_T\ttonde{x,W_1\ttonde{T}}}_{x=B\ttonde{T}} = E\quadre{e_T\ttonde{B\ttonde{T},W_1\ttonde{T}}\lvert B\ttonde{T}}.
    \end{cases}
\end{align}
We have that 
\begin{align*}
    \Bar{Y}_N\longrightarrow Y\qquad\text{in }\mathcal{S}^p\text{ for each }1\leq p<\infty,
\end{align*} 
and
\begin{equation}\label{convoftheZsrepagdyn}
\begin{alignedat}{5}
    &\lambda^{\ttonde{N}} &&\longrightarrow &\ \lambda \quad&\text{in }&&H^p\qquad \text{ for each }1\leq p <\infty,\\
    &{M_N}&&\longrightarrow&\ 0\quad&\text{in }&&\mathcal{H}^p \qquad\text{ for each }1\leq p <\infty,
\end{alignedat}
\end{equation}
where 
\begin{align*}
    \dfr M_N\ttonde{t} = \frac{1}{N}\sum_{i=1}^NZ^{\ttonde{i,N}}_2\ttonde{t}\dfr W_i\ttonde{t}.
\end{align*}
\end{teo}
\begin{oss}\label{ossexplicitsollimeq}
We have an explicit formula for the solution of the \emph{limiting equation} (\ref{repagentdyn}).
Consider the process: 
\begin{align*}
    Y\ttonde{t} = -\log\big({E_t\big[{\exp\big({-E\quadre{e_T\ttonde{x,W_1\ttonde{T}}}_{x=B\ttonde{T}}}\big)}\big]}\big),
\end{align*}
and apply It\^ o's formula to obtain 
\begin{align*}
    \dfr Y\ttonde{t}=\frac{-Z\ttonde{t}\dfr B\ttonde{t}}{E_t\quadre{\exp\big({-E\quadre{e_T\ttonde{x,W_1\ttonde{T}}}_{x=B\ttonde{T}}}\big)}} + \frac{\ttonde{Z\ttonde{t}}^2\dfr t}{2\Big( E_t\quadre{\exp\big({-E\quadre{e_T\ttonde{x,W_1\ttonde{T}}}_{x=B\ttonde{T}}}\big)} \Big)^2},
\end{align*}
where $Z$ is the unique process in $H^2$, given by the martingale representation theorem, such that
\begin{align*}
    E_t\quadre{\exp\big({-E\quadre{e_T\ttonde{x,W_1\ttonde{T}}}_{x=B\ttonde{T}}}\big)}=\int_0^tZ\ttonde{s}\dfr B\ttonde{s}.
\end{align*}
Therefore, defining
\begin{align*}
    \lambda\ttonde{t} = \frac{-Z\ttonde{t}}{E_t\quadre{\exp\big({-E\quadre{e_T\ttonde{x,W_1\ttonde{T}}}_{x=B\ttonde{T}}}\big)}},
\end{align*}
the pair $\ttonde{Y,\lambda}$ solves BSDE (\ref{repagentdyn}) and using \cite[][, {Proposition 2.1}]{BRIAND20132921} the solution belongs to $\mathcal{S}^\infty\times\mathrm{bmo}$.
\end{oss}
Let's start by proving a property of the \emph{representative agent} dynamics' terminal condition
\begin{lemma}\label{convcfin}
We have that $\mathbb{P}-$a.s.
\begin{align}\label{convcfrepag}
    \Bar{Y}_N\ttonde{T}&=\frac{1}{N}\sum_{i=1}^Ne_T\ttonde{B\ttonde{T},W_i\ttonde{T}}\to E\quadre{e_T\ttonde{B\ttonde{T},W_1\ttonde{T}}\lvert B\ttonde{T}}.
\end{align}
\end{lemma}
\begin{proof}[Proof of Lemma \ref{convcfin}]
    Define
\begin{align*}
    \phi:\R\times l^2\ttonde{\N}\to\R 
\end{align*}
\begin{align*}
    \phi\ttonde{x,\graffe{y_i}_{i\in\N}}:=\mathds{1}_{\graffe{0}}\ttonde{\gamma\ttonde{x,\graffe{y_i}_{i\in\N}}}
\end{align*}
where 
\begin{align*}
\gamma\ttonde{x,\graffe{y_i}_{i\in\N}}:=\lim_{N\to\infty} \frac{\sum_{i=1}^Ne_T\ttonde{x,y_i}}{N} -
E\quadre{e_T\ttonde{x,Y_1}}
\end{align*}
where $Y_1$ is a random variable distributed as $\mathcal{N}\ttonde{0,T}$.
Now we have that the function 
\begin{align*}
    \psi\ttonde{x}= E \phi\ttonde{x,\graffe{W_i\ttonde{T}}_{i\in\N}}=\mathbb{P}\Big({\lim_{N\to\infty} \frac{\sum_{i=1}^Ne_T\ttonde{x,W_i\ttonde{T}}}{N} =
E\quadre{e_T\ttonde{x,Y_1}}}\Big)\equiv 1
\end{align*}
for each $x\in\R$, for the strong LLN.
Thus we have 
${\psi\ttonde{B\ttonde{T}}}=1$.
We conclude that
\begin{align*}
    \mathbb{P}\Big({\lim_{N\to\infty}\frac{\sum_{i=1}^Ne_T\ttonde{B\ttonde{T},W_i\ttonde{T}}}{N} =E\quadre{e_T\ttonde{B\ttonde{T},W_1\ttonde{T}}\lvert B\ttonde{T}}\Big\lvert B\ttonde{T}}\Big)={\psi\ttonde{B\ttonde{T}}}=1
\end{align*}
where the first equality is given by \cite[][,A.3 Lemma]{schilling2012brownian}.
Finally, integrating with respect to the law of $B\ttonde{T}$ we get: 
\begin{align*}
    \mathbb{P}\Big({\lim_{N\to\infty}\frac{\sum_{i=1}^Ne_T\ttonde{B\ttonde{T},W_i\ttonde{T}}}{N} =E\quadre{e_T\ttonde{B\ttonde{T},W_1\ttonde{T}}\lvert B\ttonde{T}}}\Big)=1.
\end{align*}
\end{proof}
\begin{proof}[Proof of Theorem \ref{limrepagdyn}]
We split the proof into $4$ steps.
\begin{enumerate}[label=\textit{Step }\arabic*.,leftmargin=3.2\parindent]
\item \textit{Definition of the dynamic of the difference.}\\
To show that $\Bar{Y}_N$ converges to the solution $Y$ of (\ref{repagentdyn}) as $N\to\infty$ we consider the dynamic of the difference $\tilde{Y}\ttonde{t}:=Y\ttonde{t}-\Bar{Y}_N\ttonde{t}$,
\begin{align}\label{dyndiffsyst}
    \begin{cases}
    \dfr \Tilde{Y}\ttonde{t} &= -\frac{1}{2}\Delta^{(N)}_{\lambda}(t)\big( \Delta^{(N)}_{\lambda}(t) + 2\lambda(t) \big)\dfr t - \Delta^{(N)}_{\lambda}(t)\dfr B(t)-\dfr M_N(t), \qquad t\in [0,T)\\
    \Tilde{Y}\ttonde{T}&= E\quadre{e_T\ttonde{x,W_T}}_{x=B\ttonde{T}} - \frac{1}{N}\sum_{i=1}^Ne_T\ttonde{B\ttonde{T},W_i\ttonde{T}}.
    \end{cases}
\end{align}
Where $\Delta^{(N)}_{\lambda} = \lambda^{\ttonde{N}} - \lambda $.
The fact that the terminal condition of $\tilde{Y}$ is getting smaller and smaller as $N\to\infty$ is the key idea to prove the convergence.
Let's define 
\begin{align*}
    \ttonde{k\ttonde{N}}\ttonde{t}:=\frac{1}{2}\ttonde{\lambda\ttonde{t}+\lambda^{\ttonde{N}}\ttonde{t}},
\end{align*}
and show that there exists a

\medskip

\item\label{bmounifboundKN} \textit{bmo uniform (in $N$) bound of $k\ttonde{N}$.}\\
This boils down to show the existence of an independent (in $N$) bmo bound for $\lambda^{\ttonde{N}}$.
Let's first show that we have a bound for $\norm{\bar{Y}_N}_\infty$ independent from $N$.
A simple calculation shows that the dynamic of $\bar{Y}$ is 
\begin{align}\label{mediadeiceq}
\begin{cases}
    \dfr \Bar{Y}_N\ttonde{t} &= \frac{1}{2}\ttonde{\lambda^{\ttonde{N}}\ttonde{t}}^2\dfr t +{\lambda^{\ttonde{N}}\ttonde{t}}\dfr B\ttonde{t} + \dfr M_N\ttonde{t}, \qquad t\in [0,T)\\
    \Bar{Y}_N\ttonde{T}&=\frac{1}{N}\sum_{i=1}^Ne_T\ttonde{B\ttonde{T},W_i\ttonde{T}}\in L^\infty\ttonde{\Omega, \mathcal{F}_T, \mathbb{P}}.
    \end{cases}
\end{align}
Define $b\ttonde{t}:=\frac{1}{2}\lambda^{\ttonde{N}}\ttonde{t}$ and observe that being $\lambda^{\ttonde{N}}\in\mathrm{bmo}$ the Girsanov theorem ensures that $B^b\ttonde{t}:=B\ttonde{t}+\int_0^tb\ttonde{s}\dfr s$, $t\in\quadre{0,T}$, is a Brownian motion under the equivalent probability $\mathbb{P}^b$ such that $\frac{\dfr \mathbb{P}^b}{\dfr \mathbb{P}} = \mathcal{E}\ttonde{-\int_0^\cdot b\ttonde{s}\dfr B\ttonde{s}}_T$.
Now considering the dynamic of $\bar{Y}_N$ under $\mathbb{P}^b$ and taking the conditional expected value we get $\bar{Y}_N\ttonde{t}= E_t^{\mathbb{P}^b}\quadre{\bar{Y}_N\ttonde{T}}$, thus
\begin{align*}
    \abs{\bar{Y}_N\ttonde{t}}\leq \norm{\bar{Y}_N\ttonde{T}}_\infty\leq C
\end{align*}
where $C$ is such that $e_T\ttonde{x,y}\leq C$ for all $x,y\in\R$.
The arbitrariness of $t\in\quadre{0,T}$ gives us the result.
Consider now $\bar{Y}_N$ and take the conditional expectation under $\mathbb{P}$ to get 
\begin{align*}
    \infty>C\geq E_t\quadre{\Bar{Y}_N\ttonde{T}} \geq -\Bar{C} + \frac{1}{2}E_t\quadre{\int_t^T \lambda\ttonde{\mu^N_s}^2\dfr s} 
\end{align*}
where $\Bar{C}>{\sup_{t\in\quadre{0,T}}\abs{Y\ttonde{t}}}$ {a.s.}
We have a bound ($2\ttonde{C+\Bar{C}}$) independent from $N$ and $t$ of $E_t\quadre{\int_t^T\lambda\ttonde{\mu^N_s}^2\dfr s}$ which gives us the thesis.

\medskip

\item \textit{Convergence of} $Y$.\\
Under the measure $\mathbb{Q}^N$ such that $\frac{\dfr \mathbb{Q}^N}{\dfr \mathbb{P}}=\mathcal{E}\ttonde{-k\ttonde{N}}_T:= \mathcal{E}\ttonde{\int_0^\cdot-\ttonde{k\ttonde{N}}\ttonde{t}\dfr B\ttonde{t}}_T$ we can rewrite (\ref{dyndiffsyst}) as 
\begin{align}\label{dindiffconv}
    \begin{cases}
    \dfr \Tilde{Y}\ttonde{t} &= Z\ttonde{t}\dfr B^{\ttonde{N}}(t) - \dfr M_N\ttonde{t}, \qquad t\in [0,T)\\
    \Tilde{Y}\ttonde{T}&= E\quadre{e_T\ttonde{x,W\ttonde{T}}}_{x=B\ttonde{T}} - \frac{1}{N}\sum_{i=1}^Ne_T\ttonde{B\ttonde{T},W_i\ttonde{T}}.
    \end{cases}
\end{align}
Where $B^{\ttonde{N}}\ttonde{t}\triangleq B\ttonde{t} + \int_0^t \ttonde{k\ttonde{N}}\ttonde{s}\dfr s$ and $Z\ttonde{t}={\lambda\ttonde{t}-\lambda^{\ttonde{N}}\ttonde{t}}$.
Since $\mathcal{E}\ttonde{-k\ttonde{N}}_T$ is the Dol\' eans exponential of a bmo process we can apply the reverse Hölder inequality (we refer the reader to \cite[][,{Theorem 3.1}]{kazamaki2006continuous}), i.e.
\begin{align}\label{rhineq}
    E_\tau\quadre{\mathcal{E}\ttonde{-k\ttonde{N}}^q_T}\leq C_q\mathcal{E}\ttonde{-k\ttonde{N}}^q_\tau,\text{ for every stopping time } \tau\leq T\ a.s.
\end{align}
where $q$ and  $C_q$ depends only on the BMO bound, and they are given by: 
\begin{align}
    &\label{defLqdepBMOnorm}q:=\phi^{-1}\ttonde{\alpha},\text{ with }\phi:q\mapsto \Big({1+\frac{1}{q^2}\log\frac{2q-1}{2q-2}}\Big)^{\frac{1}{2}}-1,\text{ for }1<q<\infty \\
    &\notag C_q:=2 \Big({1-\frac{2q-2}{2q-1}\exp\big({q^2\ttonde{\alpha^2 + 2\alpha}}\big)}\Big)^{-1}.
\end{align}
Where $\alpha:=\norm{\int_0^\cdot -k\ttonde{N}_s\dfr B\ttonde{s}}_{BMO}$.
It is easy to see using the classical H\"older inequality that (\ref{rhineq}) holds for each $q'$ such that $1<q'<q$.
In particular $\mathcal{E}\ttonde{-k\ttonde{N}}_T\in L_q$.
Having a bmo bound for $-k\ttonde{N}$ independent from $N$, we conclude that the constants of the reverse Hölder inequality (\ref{rhineq}) does not depend on $N$.
Now consider a solution of (\ref{dindiffconv}) and take first the conditional expectation under $\mathbb{Q}^N$ and then the modulus to get 
\begin{align*}
    \lvert{\tilde{Y}\ttonde{t}}\lvert &\leq E_t^{\mathbb{Q}^N}\quadre{\lvert{\tilde{Y}\ttonde{T}}\lvert} = \frac{1}{E_t\quadre{\mathcal{E}\ttonde{-k\ttonde{N}}_T}}E_t\quadre{\lvert{\tilde{Y}\ttonde{T}}\lvert\mathcal{E}\ttonde{-k\ttonde{N}}_T}\\
    &\leq \mathcal{E}\ttonde{-k\ttonde{N}}_t^{-1} E_t\quadre{\mathcal{E}\ttonde{-k\ttonde{N}}^q_T}^{\frac{1}{q}}E_t\quadre{\lvert \tilde{Y}\ttonde{T}\lvert^p}^{\frac{1}{p}}\\
    &\leq C_q^{\frac{1}{q}} E_t\quadre{\lvert \tilde{Y}\ttonde{T}\lvert^p}^{\frac{1}{p}}
\end{align*}
where $\frac{1}{p}+\frac{1}{q}=1$.
Now elevating to the power of $p$ we have 
\begin{align*}
    \lvert{\tilde{Y}\ttonde{t}}\lvert^p \leq C_q^{\frac{p}{q}} E_t\quadre{\lvert \tilde{Y}\ttonde{T}\lvert^p}.
\end{align*}
Take now an arbitrary $\varepsilon>0$ and consider $\beta:=1+\frac{\varepsilon}{p}$, using Doob maximal inequality we have that
\begin{align*}
    E\quadre{\sup_{t\in\quadre{0,T}}\lvert{\tilde{Y}\ttonde{t}}\lvert^{p+\varepsilon}}\leq C E\quadre{ \sup_{t\in\quadre{0,T}}E_t\quadre{\lvert \tilde{Y}\ttonde{T}\lvert^p}^\beta}\leq CE\quadre{\lvert\tilde{Y}\ttonde{T}\lvert^{p+\varepsilon}}
\end{align*}
thus we conclude that 
\begin{align*}
    \lVert{\tilde{Y}}\lVert_{\mathcal{S}^p}\leq C_p \lVert{\tilde{Y}\ttonde{T}}\lVert_{L_p}
\end{align*}
holds true for each $p>p^*$ were $\frac{1}{p^*}+\frac{1}{q^*}=1$ and $q^*:=\phi^{-1}\ttonde{\norm{\int_0^\cdot -\ttonde{k\ttonde{N}}\ttonde{s}\dfr B\ttonde{s}}_{BMO}}$.
We conclude that
\begin{align}\label{Yvaazerodynrepag}
    \lVert \tilde{Y}\lVert_{\mathcal{S}^p}\to0\text{ for all }1\leq p<\infty.
\end{align}

\medskip

\item \textit{Convergence of the Zs.}\\
Consider a solution of (\ref{dyndiffsyst}), a direct application of It\^ o's formula provides 
\begin{align}\label{sqrtdynrepagconvthm}
    \lvert{\Tilde{Y}(0)}\lvert^2 + \int_0^T\lvert{\Delta^{(N)}_\lambda(t)}\lvert^2\dfr t + \int_0^T\dfr\cov{M_N}_t = \lvert{\Tilde{Y}(T)}\lvert^2 + 2\nu +\eta,
\end{align}
where 
\begin{align*}
    \nu=\int_0^T \Tilde{Y}(t) \big(\Delta^{(N)}_\lambda(t)\dfr B(t) + \dfr M_N(t)\big)\qquad \eta=\int_0^T \Tilde{Y}(t) \Delta^{(N)}_\lambda(t)\big( \Delta^{(N)}_\lambda(t)+2\lambda(t) \big)\dfr t.
\end{align*}
Consider $\eta$ and compute
\begin{alignat*}{2}
    \eta &\leq \int_0^T \lvert{\Tilde{Y}(t)}\lvert \big( \lvert{\Delta^{(N)}_\lambda(t)}\lvert+2\abs{\lambda(t)} \big)\lvert{\Delta^{(N)}_\lambda(t)}\lvert\dfr t\leq &&\frac{1}{2}\sup_{t\in[0,T]}\lvert \Tilde{Y}(t)\lvert^2\int_0^T \big( \lvert{\Delta^{(N)}_\lambda(t)}\lvert+2\abs{\lambda(t)} \big)^2\dfr t\\
    & &&+\frac{1}{2}\int_0^T\lvert{\Delta^{(N)}_\lambda(t)}\lvert^2\dfr t\\
\end{alignat*}
Putting this expression in (\ref{sqrtdynrepagconvthm}) we get
\begin{align*}
    \int_0^T\lvert{\Delta^{(N)}_\lambda(t)}\lvert^2\dfr t + \int_0^T\dfr\cov{M_N}_t\leq 2\lvert{\Tilde{Y}(T)}\lvert^2 + 4\nu + \sup_{t\in[0,T]}\lvert \Tilde{Y}(t)\lvert^2\int_0^T \big( \lvert{\Delta^{(N)}_\lambda(t)}\lvert+2\abs{\lambda(t)} \big)^2\dfr t.
\end{align*}
Fix an arbitrary $p>p^*$ and consider
\begin{alignat*}{2}
    &E&&\Bigg[{\Big( \int_0^T\lvert{\Delta^{(N)}_\lambda(t)}\lvert^2\dfr t + \int_0^T\dfr\cov{M_N}_t \Big)^{\frac{p}{2}}}\Bigg] \\
    &\leq  &&C_p \Bigg( \lVert{\tilde{Y}\ttonde{T}}\lVert_{L_p}^p + E\quadre{\abs{\int_0^T \Tilde{Y}(t) \big(\Delta^{(N)}_\lambda(t)\dfr B(t) + \dfr M_N(t)\big)}^{\frac{p}{2}}}\Bigg)\\
    & &&+C_p E\quadre{\sup_{t\in[0,T]}\lvert \Tilde{Y}(t)\lvert^p\Big( \int_0^T \big( \lvert{\Delta^{(N)}_\lambda(t)}\lvert+2\abs{\lambda(t)} \big)^2\dfr t  \Big)^{\frac{p}{2}}}.
\end{alignat*}
The application of Burkholder-Davis-Gundy and H\"older inequality provide
\begin{alignat*}{2}
    &\leq&& C_p \Bigg( \lVert{\tilde{Y}\ttonde{T}}\lVert_{L_p}^p + E\quadre{\sup_{t\in[0,T]}\lvert \Tilde{Y}(t)\lvert^p}^{\frac{1}{2}} E\quadre{\Big( \int_0^T\lvert{\Delta^{(N)}_\lambda(t)}\lvert^2\dfr t + \int_0^T\dfr\cov{M_N}_t \Big)^{\frac{p}{2}}}^{\frac{1}{2}}\Bigg)\\
    & &&+C_p E\quadre{\sup_{t\in[0,T]}\lvert \Tilde{Y}(t)\lvert^{2p}}^{\frac{1}{2}}E\quadre{\Big( \int_0^T \big( \lvert{\Delta^{(N)}_\lambda(t)}\lvert+2\abs{\lambda(t)} \big)^2\dfr t\Big)^p}^{\frac{1}{2}}.
\end{alignat*}
Besides, since $\Delta^{(N)}_\lambda$ and $\lambda$ belong to bmo, the energy inequality for BMO martingales, see e.g. Section 2.1 in \cite{kazamaki2006continuous}, provides 
\begin{align*}
    E&\quadre{\Big( \int_0^T \lvert{\Delta^{(N)}_\lambda(t)}\lvert^2\dfr t \Big)^p + \Big( \int_0^T 4\abs{\lambda(t)}^2 \dfr t\Big)^p }\\
    &\leq p!\Bigg( \norm{\int_0^\cdot \Delta^{(N)}_\lambda(t) \dfr B(t) }_{BMO}^{2p} + \norm{\int_0^\cdot 2\lambda(t)\dfr B(t)}_{BMO}^{2p} \Bigg)\leq L,
\end{align*}
where $L$ does not depend on $N$.
Combining the two previous expressions together with the inequality $ab\leq 2a^2 +\frac{b^2}{2}$ yields 
\begin{align*}
\frac{1}{2}\Big(\norm{\Delta^{(N)}_\lambda}_{H^p}^p + \norm{M_N}_{\mathcal{H}^p}^p\Big)&\leq\frac{1}{2} E\quadre{\Big( \int_0^T\lvert{\Delta^{(N)}_\lambda(t)}\lvert^2\dfr t + \int_0^T\dfr\cov{M_N}_t \Big)^{\frac{p}{2}}} \\
&\leq  C_p\Big( \lVert{\tilde{Y}\ttonde{T}}\lVert_{L_p}^p + \lVert \Tilde{Y} \lVert^{p}_{\mathcal{S}^p} + \lVert \Tilde{Y} \lVert^{p}_{\mathcal{S}^{2p}}\Big).
\end{align*}
Hence, the arbitrariness of $p>p^*$ together with (\ref{Yvaazerodynrepag}) leads to 
\begin{align}\label{convoftheZsrepagent}
    \norm{\Delta^{(N)}_\lambda}_{H^p}\to 0\qquad  \norm{M_N}_{\mathcal{H}^p}\to 0,\text{ for all } 1\leq p <\infty.
\end{align}
\end{enumerate}
\end{proof}

\subsection{BSDE system in the limit}
In this section we prove that as $N$ goes to infinity, conditionally on the common Brownian motion $B$, the system decouples.
We also get an explicit formula for the solution of the limit system. 
\begin{teo}
(BSDE limit system)
Suppose $\ttonde{Y,Z}$, is an $\mathcal{S}^\infty\times\mathrm{bmo}-$solution to 
\begin{align}\label{equilibriuminfagent}
    \begin{cases}
    \dfr Y\ttonde{t} &= Z_1\ttonde{t}\dfr B\ttonde{t} +Z_2\ttonde{t}\dfr W_1\ttonde{t} -\ttonde{\frac{1}{2}{\lambda^{2}\ttonde{t}} - \lambda\ttonde{t}Z_1\ttonde{t}}\dfr t,\ t\in [0,T)\\
    Y\ttonde{T}&=e_T\ttonde{B\ttonde{T},W_1\ttonde{T}}\in L^\infty\ttonde{\Omega, \mathcal{F}_T, \mathbb{P}}.
    \end{cases}
\end{align}
For a fixed $k\in\N$, we have that as $N\to\infty$
\begin{alignat}{3}\label{limitkSyst}
    \begin{aligned}
    &Y^{\ttonde{i,N}}\longrightarrow Y^{\ttonde{i}}\qquad&&\text{in }\mathcal{S}^p\text{ for each }1\leq p<\infty,\qquad &&i\in\graffe{1,\dots,k}\\
    &Z_j^{\ttonde{i,N}}\longrightarrow Z_j^{\ttonde{i}}\qquad&&\text{in }H^p\text{ for each }1\leq p<\infty,\qquad &&i\in\graffe{1,\dots,k}\qquad j=1,2
    \end{aligned}
\end{alignat}
where $({Y^{\ttonde{1}}, Z_1^{\ttonde{1}},Z_2^{\ttonde{1}}}),\dots,({Y^{
\ttonde{k}},Z_1^{\ttonde{k}},Z_2^{\ttonde{k}} })$ are independent
(conditionally on $B$) copies of the solution
of equation (\ref{equilibriuminfagent}).
\end{teo}
\begin{oss}
We interpret the limit (\ref{limitkSyst}) as saying that the agents' dynamics $Y^{\ttonde{i,N}}$ become asymptotically (conditionally on $B$) i.i.d. as $N\to\infty$.
\end{oss}
\begin{oss}
We have an explicit formula for the solution of equation (\ref{equilibriuminfagent}). 
Consider the measure $\mathbb{P}^\lambda$ such that $\frac{\dfr \mathbb{P}^\lambda}{\dfr \mathbb{P}}=\mathcal{E}({-\int_0^\cdot\lambda\ttonde{t}\dfr B\ttonde{t}})_T$ and define 
\begin{align*}
    Y\ttonde{t}:=E_t^{\mathbb{P}^\lambda}\quadre{e_T\ttonde{B\ttonde{T},W_1\ttonde{T}} + \frac{1}{2}\int_t^T\lambda\ttonde{s}^2\dfr s}.
\end{align*}
Using It\^ o's's formula it is easy to see that the unique solution to this equation is the triple $\ttonde{Y,Z_1,Z_2}$ with $Z_1,Z_2$ such that 
\begin{align*}
    \int_0^tZ_1\ttonde{s}\dfr B\ttonde{s} + \int_0^tZ_2\ttonde{s}\dfr W_1\ttonde{s} = E_t^{\mathbb{P}^\lambda}\quadre{e_T\ttonde{B\ttonde{T},W_1\ttonde{T}} + \frac{1}{2}\int_0^T\lambda\ttonde{s}^2\dfr s}.
\end{align*}
\end{oss}
\begin{proof}
Observe first that the only thing we need to prove is that for a fixed $i\in\graffe{1,\dots,k}$
\begin{alignat*}{3}
    \begin{aligned}
    &Y^{\ttonde{i,N}}\longrightarrow Y^{\ttonde{i}}\qquad&&\text{in }\mathcal{S}^p\text{ for each }1\leq p<\infty\\
    &Z_j^{\ttonde{i,N}}\longrightarrow Z_j^{\ttonde{i}}\qquad&&\text{in }H^p\text{ for each }1\leq p<\infty,\qquad j=1,2.
    \end{aligned}
\end{alignat*}
We split the proof into $4$ steps. 

\medskip

\begin{enumerate}[label=\textit{Step }\arabic*.,leftmargin=3.2\parindent]
\item \textit{Definition of the dynamic of the difference.}\\
To show that $Y^{\ttonde{i,N}}\longrightarrow Y^{\ttonde{i}}$ as $N$ goes to infinity we consider the dynamic of the difference $D\ttonde{t}:=Y^{\ttonde{i}}\ttonde{t}-Y^{\ttonde{i,N}}\ttonde{t}$,
\begin{align}\label{dyndiffBSDE}
    \begin{cases}
    \dfr D\ttonde{t} &= ({Z_1^{\ttonde{i}}\ttonde{t}-Z_1^{\ttonde{i,N}}\ttonde{t}})\dfr B\ttonde{t} + (Z^{\ttonde{i}}_2\ttonde{t}  -Z_2^{\ttonde{i,N}}\ttonde{t})\dfr W_i\ttonde{t} \\
     &\quad-\Big({\frac{1}{2}\big({\lambda^{2}\ttonde{t} - (\lambda^{\ttonde{N}}\ttonde{t})^2}\big) - \big(\lambda\ttonde{t}Z^{\ttonde{i}}_1\ttonde{t} - \lambda^{\ttonde{N}}\ttonde{t}Z^{\ttonde{i,N}}_1\ttonde{t}\big)}\Big)\dfr t,\ t\in [0,T)\\
    D\ttonde{T}&=0.
\end{cases}
\end{align}

\medskip

\item \textit{bmo uniform bound for} $\lambda^{\ttonde{N}}$.\\
We already obtained this in {\ref{bmounifboundKN}} of the proof of {Theorem \ref{limrepagdyn}}.

\medskip

\item \textit{Convergence of} $Y^{\ttonde{i,N}}$.\\
Under the measure $\mathbb{Q}^N$ such that $\frac{\dfr \mathbb{Q}^N}{\dfr \mathbb{P}}=\mathcal{E}\ttonde{\int_0^\cdot-\lambda^{\ttonde{N}}\ttonde{t}\dfr B\ttonde{t}}_T$ we can rewrite (\ref{dyndiffBSDE}) as 
\begin{align}\label{BSDEdiffsystmeaschange}
    \begin{cases}
    \dfr D\ttonde{t} &= ({Z_1^{\ttonde{i}}\ttonde{t}-Z_1^{\ttonde{i,N}}\ttonde{t}})\dfr B^{\ttonde{N}}\ttonde{t} + (Z^{\ttonde{i}}_2\ttonde{t}  -Z_2^{\ttonde{i,N}}\ttonde{t})\dfr W_i\ttonde{t} \\
     &\quad-\Big({\frac{1}{2}\big({\lambda^{2}\ttonde{t} - (\lambda^{\ttonde{N}}\ttonde{t})^2}\big) - 
     Z^{\ttonde{i}}_1\ttonde{t}\big({\lambda\ttonde{t}-\lambda^{\ttonde{N}}\ttonde{t}}\big)}\Big)\dfr t,\ t\in [0,T)\\
    D\ttonde{T}&=0.
\end{cases}
\end{align}
Where $B^{\ttonde{N}}\triangleq B\ttonde{t} + \int_0^t\lambda^{N}\ttonde{s}\dfr s$.
Applying the same reasoning as in the proof of {Theorem \ref{limrepagdyn}} we get that there exists a $C_q,q$ independent from $N$ such that 
\begin{align}
    E_\tau\quadre{\mathcal{E}\big({-\lambda^{\ttonde{N}}}\big)_T^q}\leq C_q\mathcal{E}\big({-\lambda^{\ttonde{N}}}\big)_\tau^q,\text{ for every stopping time } \tau\leq T\ a.s.
\end{align}
Where $\mathcal{E}\ttonde{-\lambda^{\ttonde{N}}}_t=\mathcal{E}\ttonde{\int_0^\cdot -\lambda^{\ttonde{N}}\ttonde{s}\dfr B\ttonde{s}}_t$.
It is easy to see using the classical H\"older inequality that this inequality holds for each $q'$ such that $1<q'<q$.
In particular $\mathcal{E}\ttonde{-k\ttonde{N}}_T\in L_q$.
Now consider a solution of (\ref{BSDEdiffsystmeaschange}) and take first the
conditional expectation under $\mathbb{Q}^N$ and then the modulus to get
\begin{align}
    \notag\abs{D\ttonde{t}}&\leq E_t^{\mathbb{Q}^N}\quadre{\int_t^T\abs{\Delta^{\ttonde{N}}_\lambda\ttonde{s}}\abs{\lambda\ttonde{s}+\lambda^{\ttonde{N}}\ttonde{s}}\dfr s} + E_t^{\mathbb{Q}^N}\quadre{\int_t^T \abs{Z_1^{\ttonde{i}}\ttonde{s}}\abs{\Delta^{\ttonde{N}}_\lambda\ttonde{s}}\dfr s}\\
    &\notag\leq E^{\mathbb{Q}^N}_t\quadre{\int_0^T\abs{\Delta^{\ttonde{N}}_\lambda\ttonde{s}}^2\dfr s}^{\frac{1}{2}} \Bigg(E_t^{\mathbb{Q}^N}\quadre{\int_t^T\abs{Z^{\ttonde{i}}_1\ttonde{s}}^2\dfr s}^{\frac{1}{2}}\\ &\notag\quad+E_t^{\mathbb{Q}^N}\quadre{\int_t^T\abs{\lambda\ttonde{s}+\lambda^{\ttonde{N}}\ttonde{s}}^2\dfr s}^{\frac{1}{2}} \Bigg)\\
    &\label{ineqbmonormequivlenceunderunifbound}\leq C E^{\mathbb{Q}^N}_t\quadre{\int_0^T\abs{\Delta^{\ttonde{N}}_\lambda\ttonde{s}}^2\dfr s}^{\frac{1}{2}}\Big({\norm{Z^{\ttonde{i}}_1}_{\mathrm{bmo}} + \norm{\lambda+\lambda^{\ttonde{N}}}_{\mathrm{bmo}} }\Big)\\
    &\notag\leq C E^{\mathbb{Q}^N}_t\quadre{\int_0^T\abs{\Delta^{\ttonde{N}}_\lambda\ttonde{s}}^2\dfr s}^{\frac{1}{2}} \\
    &\notag= \frac{C}{E_t\quadre{\mathcal{E}\ttonde{-\lambda^{\ttonde{N}}}_T}^{\frac{1}{2}}}E_t\quadre{\mathcal{E}\ttonde{-\lambda^{\ttonde{N}}}_T\int_0^T\abs{\Delta^{\ttonde{N}}_\lambda\ttonde{s}}^2\dfr s}^{\frac{1}{2}}
\end{align} where $\Delta^{\ttonde{N}}_\lambda = \lambda-\lambda^{\ttonde
 {N}} $ and inequality (\ref{ineqbmonormequivlenceunderunifbound}) is
 given by the existence of a uniform bmo bound for
 $k\ttonde{N}$ and \cite[][,{Theorem 3.6}]{kazamaki2006continuous} which gives us the existence of $A,B$ such that for each $N\in\N$
 \begin{align*}
     A\norm{\cdot}_{BMO,\mathbb{Q}^N}\leq \norm{\cdot}_{BMO,\mathbb{P}}\leq B\norm{\cdot}_{BMO,\mathbb{Q}^N},
 \end{align*}
 where $A$ and $B$ depend only on the uniform bmo bound for $k\ttonde{N}$, thus making all the the BMO norms taken with respect to the probability measure $\mathbb{Q}^N$ equivalent to the BMO norm take with respect to $\mathbb{P}$.
Now taking the square on both sides we get
\begin{align*}
    \abs{D\ttonde{t}}^2&\leq \frac{C^2}{E_t\quadre{\mathcal{E}\ttonde{-\lambda^{\ttonde{N}}}_T}}E_t\quadre{\mathcal{E}\ttonde{-\lambda^{\ttonde{N}}}_T\int_0^T\abs{\Delta^{\ttonde{N}}_\lambda\ttonde{s}}^2\dfr s}\\
    &\leq C^2\mathcal{E}\ttonde{-\lambda^{\ttonde{N}}}_t^{-1}E_t\quadre{\mathcal{E}\ttonde{-\lambda^{\ttonde{N}}}_T^q}^{\frac{1}{q}}E_t\quadre{\Big({\int_0^T\abs{\Delta^{\ttonde{N}}_\lambda\ttonde{s}}^{2}\dfr s}\Big)^p}^{\frac{1}{p}}\\
    &\leq C^2C_q^{\frac{1}{q}}E_t\quadre{\Big({\int_0^T\abs{\Delta^{\ttonde{N}}_\lambda\ttonde{s}}^{2}\dfr s}\Big)^p}^{\frac{1}{p}}
\end{align*}
where $\frac{1}{p}+\frac{1}{q}=1$.
Now elevating to the power of $p$ we have 
\begin{align*}
    \abs{D\ttonde{t}}^{2p}\leq C E_t\quadre{\Big({\int_0^T\abs{\Delta^{\ttonde{N}}_\lambda\ttonde{s}}^{2}\dfr s}\Big)^p}
\end{align*}
Take now an arbitrary $\varepsilon>0$ and consider $\beta:=1+\frac{\varepsilon}{2p}$, using Doob maximal inequality we have that
\begin{align*}
    E\quadre{\sup_{t\in\quadre{0,T}}\abs{D\ttonde{t}}^{2p+\varepsilon}}&\leq C E\quadre{\sup_{t\in\quadre{0,T}}E_t\quadre{\Big({\int_0^T\abs{\Delta^{\ttonde{N}}_\lambda\ttonde{s}}^{2}\dfr s}\Big)^p}^\beta}\\
    &\leq C E\quadre{\Big({\int_0^T\abs{\Delta^{\ttonde{N}}_\lambda\ttonde{s}}^{2}\dfr s}\Big)^{p+ \frac{\varepsilon}{2}}}
\end{align*}
thus we conclude that 
\begin{align*}
    \norm{D}_{\mathcal{S}^{2p}}\leq C_p\norm{\Delta_\lambda^{\ttonde{N}}}_{H^{2p}}
\end{align*}
holds true for each $p>p^*$ were $\frac{1}{p^*}+\frac{1}{q^*}=1$ and $q^*:=\phi^{-1}\ttonde{\norm{\int_0^\cdot -\lambda^{\ttonde{N}}\ttonde{s}\dfr B\ttonde{s}}_{BMO}}$ where $\phi$ is defined in (\ref{defLqdepBMOnorm}).
Now (\ref{convoftheZsrepagdyn}) let us conclude that
\begin{align}\label{Difftendeazero}
    \lVert D\lVert_{\mathcal{S}^p}\to0\text{ for all }1\leq p<\infty
\end{align}

\medskip

\item \textit{Convergence of $Z_j^{\ttonde{i,N}}$, j=1,2}.\\
Define $\Delta_{Z}^{\ttonde{N,j}}:=\ttonde{Z_j^{\ttonde{i}} -Z_j^{\ttonde{i,N}} }$ for $j=1,2$.
Consider a solution of (\ref{BSDEdiffsystmeaschange}) and the dynamic of its square: 
\begin{align}\label{sqrtdynofdiffBSDEsystlimit}
    \dfr \abs{D\ttonde{t}}^2 = 2D\ttonde{t}\dfr D\ttonde{t} +\dfr  \cov{D}_t = 2D\ttonde{t}\dfr M\ttonde{t} - 2D\ttonde{t}\chi\ttonde{t}\dfr t + \dfr \cov{D}_t
\end{align}
where 
\begin{align*}
    \dfr M\ttonde{t}&=(\Delta_{Z}^{\ttonde{N,1}}(t))\dfr B\ttonde{t} + (\Delta_{Z}^{\ttonde{N,2}}(t))\dfr W_i\ttonde{t} \\
    \chi\ttonde{t}&=\Big({\frac{1}{2} \Delta^{(N)}_\lambda(t)\big({\lambda\ttonde{t} + \lambda^{\ttonde{N}}\ttonde{t}}\big) - \big( \Delta^{(N)}_\lambda(t)Z^{(i)}_1(t) +\lambda^{(N)}(t)\Delta^{(N,1)}_{Z}(t) \big) }\Big).
\end{align*}
Consider the integral form of (\ref{sqrtdynofdiffBSDEsystlimit})
\begin{align}\label{sqrtdynofdiffBSDEsystlimitintform}
    \abs{D\ttonde{0}}^2 + \int_0^T\dfr\cov{D}_t = -\int_0^T2D\ttonde{s}\dfr M\ttonde{s} + \int_0^T2D\ttonde{s}\chi\ttonde{s}\dfr s =: \iota + \zeta 
\end{align}
Focusing on $\zeta$ we compute
\begin{alignat*}{2}
    \zeta & \leq && 2\int_0^T\lvert D\ttonde{t} \lvert \big( \lvert \lambda(t) \lvert + \lvert \lambda^{(N)}(t)\lvert + \lvert Z^{(i)}_1(t) \lvert \big)\lvert \Delta^{(N)}_\lambda(t) \lvert \dfr t + 2 \int_0^T \lvert D\ttonde{t} \lvert \lvert \lambda^{(N)}(t)\lvert \lvert\Delta^{(N,1)}_{Z}(t) \lvert \dfr t\\
    &\leq && 2\sup_{t\in[0,T]}\abs{D(t)}^2 \Big( \int_0^T \big( \lvert \lambda(t) \lvert + \lvert \lambda^{(N)}(t)\lvert + \lvert Z^{(i)}_1(t) \lvert \big)^2 + \lvert \lambda^{(N)}(t)\lvert^2 \dfr t \Big) + \frac{1}{2}\int_0^T \lvert \Delta^{(N)}_\lambda(t) \lvert^2 \dfr t  \\
    & &&+\frac{1}{2}\int_0^T \lvert\Delta^{(N,1)}_{Z}(t) \lvert^2 \dfr t.
\end{alignat*}
Putting this expression in (\ref{sqrtdynofdiffBSDEsystlimitintform}), we get 
\begin{alignat*}{1}
    &\int_0^T \big(\big({\Delta_Z^{\ttonde{N,1}}\ttonde{t}}\big)^2 +  \big({\Delta_Z^{\ttonde{N,2}}\ttonde{t}}\big)^2\big)\dfr t \\
    &\leq 2\abs{\iota} + 4\sup_{t\in[0,T]}\abs{D(t)}^2 \Big( \int_0^T \big( \lvert \lambda(t) \lvert + \lvert \lambda^{(N)}(t)\lvert + \lvert Z^{(i)}_1(t) \lvert \big)^2 + \lvert \lambda^{(N)}(t)\lvert^2 \dfr t \Big)  + \int_0^T \lvert \Delta^{(N)}_\lambda(t) \lvert^2 \dfr t.
\end{alignat*}
Fix an arbitrary $p>p^*$ and consider 
\begin{alignat*}{2}
    &E&&\Big[{\Big(\int_0^T\big( \big({\Delta_Z^{\ttonde{N,1}}\ttonde{t}}\big)^2 +  \big({\Delta_Z^{\ttonde{N,2}}\ttonde{t}}\big)^2\big)\dfr t \Big)^{\frac{p}{2}}}\Big]  \\
    &\leq &&C_p E\quadre{\abs{\iota}^{\frac{p}{2}}} + C_pE\quadre{\sup_{t\in[0,T]}\abs{D(t)}^p\Big( \int_0^T \big( \lvert \lambda(t) \lvert + \lvert \lambda^{(N)}(t)\lvert + \lvert Z^{(i)}_1(t) \lvert \big)^2 + \lvert \lambda^{(N)}(t)\lvert^2 \dfr t \Big)^{\frac{p}{2}}} \\
    & &&+ C_p\lVert \Delta^{(N)}_\lambda\lVert_{H^p}^p\\
    &\leq && C_p E\quadre{\abs{\iota}^{\frac{p}{2}}} + C_p\norm{D}^p_{\mathcal{S}^{2p}} E\quadre{\Big( \int_0^T \big( \lvert \lambda(t) \lvert + \lvert \lambda^{(N)}(t)\lvert + \lvert Z^{(i)}_1(t) \lvert \big)^2 + \lvert \lambda^{(N)}(t)\lvert^2 \dfr t \Big)^p}^{\frac{1}{2}} \\
    & &&+ C_p\lVert \Delta^{(N)}_\lambda\lVert_{H^p}^p.
\end{alignat*}
Since $\lambda$, $\lambda^{(N)}$ and $Z^{(i)}_1$ belongs to the bmo space, we can apply the energy inequality for BMO martingales as in the proof of Theorem \ref{limrepagdyn} and bound the previous expression with
\begin{align*}
    C\Big( E\quadre{\abs{\iota}^{\frac{p}{2}}} + \norm{D}^p_{\mathcal{S}^{2p}} + \lVert \Delta^{(N)}_\lambda\lVert_{H^p}^p \Big),
\end{align*}
where $C$ is independent from $N$.
Consider now the first addendum of the previous expression, H\"older and Burkholder-Davis-Gundy inequality provide
\begin{align*}
    CE\quadre{\abs{\iota}^{\frac{p}{2}}}&\leq C E\quadre{\sup_{t\in[0,T]}\abs{D(t)}^p}^{\frac{1}{2}}E\Big[{\Big(\int_0^T\big( \big({\Delta_Z^{\ttonde{N,1}}\ttonde{t}}\big)^2 +  \big({\Delta_Z^{\ttonde{N,2}}\ttonde{t}}\big)^2\big)\dfr t \Big)^{\frac{p}{2}}}\Big]^{\frac{1}{2}}\\
    &\leq 2C^2\norm{D}_{\mathcal{S}^p}^p + \frac{1}{2}E\Big[{\Big(\int_0^T\big( \big({\Delta_Z^{\ttonde{N,1}}\ttonde{t}}\big)^2 +  \big({\Delta_Z^{\ttonde{N,2}}\ttonde{t}}\big)^2\big)\dfr t \Big)^{\frac{p}{2}}}\Big],
\end{align*}
for a possibly larger constant $C$, observe that in the last inequality we have used the inequality $ab<2a^2 + \frac{b^2}{2}$.
Putting all the estimates together we conclude that 
\begin{align*}
    \frac{1}{2}\big( \lVert{\Delta_Z^{(N,1)}}\lVert_{H^p}^p +  \lVert{\Delta_Z^{(N,2)}}\lVert_{H^p}^p \big)&\leq \frac{1}{2} E\Big[{\Big(\int_0^T\big( \big({\Delta_Z^{\ttonde{N,1}}\ttonde{t}}\big)^2 +  \big({\Delta_Z^{\ttonde{N,2}}\ttonde{t}}\big)^2\big)\dfr t \Big)^{\frac{p}{2}}}\Big]  \\
    &\leq C \big( \norm{D}_{\mathcal{S}^p}^p +\norm{D}^p_{\mathcal{S}^{2p}} + \lVert \Delta^{(N)}_\lambda\lVert_{H^p}^p\big).
\end{align*}
Hence, the arbitrariness of $p>p^*$ together with (\ref{Difftendeazero}) and (\ref{convoftheZsrepagent}) leads to 
\begin{align*}
    \lVert{\Delta_Z^{(N,1)}}\lVert_{H^p}\to 0\qquad  \lVert{\Delta_Z^{(N,2)}}\lVert_{H^p}\to 0,\text{ for all } 1\leq p <\infty, \text{ as }N\to\infty
\end{align*}
\end{enumerate}
\end{proof}

\subsection{Existence of an equilibrium}
In this section we show that equation (\ref{equilibriumNagent}) has a unique solution in the class of bounded continuous solution.
\begin{defn}\label{defsolmark}
(A Markovian solution to BSDE).
Given borel function $\boldsymbol{f}:\R^{N\times d}\to\R^N$ and $\boldsymbol{g}:\R^{d}\to\R^N$, a pair $\ttonde{\boldsymbol{v},\boldsymbol{w}}$ of Borel functions with the domain $\quadre{0,T}\times\R^d$ and co-domains $\R^N$ and $\R^{N\times d}$, respectively, is called a \textbf{Markovian solution} to the \textbf{system}
\begin{align}\label{defmarksol}
    \dfr\boldsymbol{Y}\ttonde{t} = \boldsymbol{Z}\ttonde{t}\dfr \boldsymbol{W}\ttonde{t} - \boldsymbol{f}\ttonde{\boldsymbol{Z}\ttonde{t}}\dfr t,\qquad \boldsymbol{Y}\ttonde{T} = \boldsymbol{g}\ttonde{\boldsymbol{W}\ttonde{T}},
\end{align}
of \textbf{backward stochastic differential equations} if
\begin{enumerate}
    \item $\boldsymbol{Y}:= \boldsymbol{v}\ttonde{\cdot, \boldsymbol{W}}$ is a continuous process, $\boldsymbol{Z}:= \boldsymbol{w}\ttonde{\cdot, \boldsymbol{W}}\in\mathcal{P}^2$, and $\boldsymbol{f}\ttonde{\boldsymbol{Z}}\in\mathcal{P}^1$,
    \item For all $t\in\quadre{0,T}$, we have 
    \begin{align*}
        \boldsymbol{Y}\ttonde{t} = \boldsymbol{g}\ttonde{\boldsymbol{W}\ttonde{T}} + \int_{t}^T \boldsymbol{f}\ttonde{\boldsymbol{Z}\ttonde{u}}\dfr u - \int_{t}^T\boldsymbol{Z}\ttonde{u}\dfr \boldsymbol{W}_u,\ \mathrm{a.s.}
    \end{align*}
\end{enumerate}
A Markovian solution $\ttonde{\boldsymbol{v},\boldsymbol{w}}$ to (\ref{defmarksol}) is said to be \textbf{bounded} if $\boldsymbol{v}$ is bounded, \textbf{continuous} if $\boldsymbol{v}$ is continuous and a \textbf{bmo-solution} if $\boldsymbol{Z}\in\mathrm{bmo}$.
\end{defn}
Let's now state technical conditions useful to prove the existence of a
solution to systems as in (\ref{defmarksol}).
We will define them as \say{simplified}, since we are using a less general version of the analogous conditions in \cite{XinZit18}, we refer the reader to this paper for the general version.
\begin{defn}
(The simplified Bensoussan-Frehse (sBF) condition).
We say that a continuous function $\boldsymbol{f}:\R^{N\times d}\to\R^N$ satisfies the \textbf{condition} \textbf{(}${\boldsymbol{\mathrm{sBF}}}$\textbf{)} if it admits a decomposition of the form 
\begin{align}
    \boldsymbol{f}\ttonde{\boldsymbol{z}} = \mathrm{diag}\ttonde{\boldsymbol{z}\boldsymbol{I}\ttonde{\boldsymbol{z}}} + \boldsymbol{q}\ttonde{\boldsymbol{z}},
\end{align}
such that the function $\boldsymbol{I}:\R^{N\times d}\to\R^{N\times d}$ and $\boldsymbol{q}:\R^{N\times d}\to\R^N$ have the following property: there exists a positive constant ${C}$ such that for all $\boldsymbol{z}\in\R^{N\times d}$ we have
\begin{alignat*}{3}
    &\abs{\boldsymbol{I}\ttonde{\boldsymbol{z}}}&&\leq C\ttonde{ 1 + \abs{\boldsymbol{z}}} &&{\text{(quadratic-linear)}}\\
    &\abs{\boldsymbol{q}^i\ttonde{\boldsymbol{z}}}&&\leq C\bigg({ 1 + \sum_{j=1}^i\abs{\boldsymbol{z}^j}^2}\bigg),\qquad i=1,\dots,N,\qquad &&\text{(quadratic-triangular)}
\end{alignat*}
In that case, we write $\boldsymbol{f}\in\boldsymbol{\mathrm{sBF}}$(${C}$)
\end{defn}
Another ingredient necessary to guarantee the existence of a solution to (\ref{defmarksol}) is a-priori boundedness.
We remind the reader that a set of non-zero vectors $\boldsymbol{a}_1,\dots,\boldsymbol{a}_K$ in $\R^N$ (with $K>N$) is said to \textbf{positively span} $\R^N$, if, for each $\boldsymbol{a}\in\R^N$ there exists nonnegative constants $\lambda_1,\dots,\lambda_K$ such that 
\begin{align*}
    \lambda_1\boldsymbol{a}_1 + \dots + \lambda_K\boldsymbol{a}_K = \boldsymbol{a}.
\end{align*}
The following two well-known characterization, presented here for reader convenience, make positively spanning sets easy to spot: ($1$) Non-zero vectors $\boldsymbol{a}_1,\dots, \boldsymbol{a}_K$ positively span $\R^N$ if for every $\boldsymbol{a}\in\R^N\setminus\graffe{\boldsymbol{0}}$ there exists $k\in\graffe{1,\dots,K}$ such that $\boldsymbol{a}^T\boldsymbol{a}_k>0$.
($2$) If non-zero vectors $\boldsymbol{a}_1,\dots,\boldsymbol{a}_k$ already span $\R^N$, then they positively span $\R^N$ if $\boldsymbol{0}$ admits a nontrivial positive representation, i.e., if there exist nonnegative $\lambda_1,\dots,\lambda_K$, not all $0$, such that $\lambda_1\boldsymbol{a}_1+\dots+\lambda_K\boldsymbol{a}_K=\boldsymbol{0}$.
\begin{defn}
(The simplified a-priori boundedness (sAB) condition).
We say that $\boldsymbol{f}$ satisfies the \textbf{condition (sAB)} if there exist a set $\boldsymbol{a}_1,\dots,\boldsymbol{a}_K$ which positively spans $\R^N$, such that
\begin{align}
    \boldsymbol{a}^T_k\boldsymbol{f}\ttonde{\boldsymbol{z}}\leq \frac{1}{2}\abs{\boldsymbol{a}^T_k\boldsymbol{z}}^2,\qquad\text{for all}\ \boldsymbol{z}\ \text{and }k=1,\dots,K.
\end{align}
\end{defn}
\begin{oss}
Condition (sAB) is invariant under linear invertible transformation of $\R^N$.
\end{oss}
\begin{teo}\label{solexist}
The system (\ref{equilibriumNagent}) admits a unique ${\mathcal{S}^\infty\times\mathrm{bmo}}-$solution of the form 
\begin{align*}
    \boldsymbol{Y}:= \boldsymbol{v}\ttonde{\cdot, \boldsymbol{W}},\qquad \boldsymbol{Z}:= \boldsymbol{w}\ttonde{\cdot, \boldsymbol{W}}
\end{align*}
with $\ttonde{t,x}\in\quadre{0,T}\times\R^{N+1}$, and $\ttonde{\boldsymbol{v},\boldsymbol{w}}$ is a \emph{\textbf{bounded continuous bmo-solution}} as in \emph{{Definition \ref{defsolmark}}}.
\end{teo}

\begin{proof}
We will rewrite \cite[][{Theorem 2.14}]{XinZit18} adapting its statement to this context as follows
\begin{teo}\label{exandunsolsyst}
(Existence under (sBF)$+$(sAB))
Suppose that $\boldsymbol{f}$ satisfies conditions \emph{(sBF)} and \emph{(sAB)} and that $\boldsymbol{g}$ is bounded and $\alpha-$Hölder continuous for some $\alpha>0$.
Then the system \emph{(\ref{defmarksol})} admit a bounded continuous bmo-solution $\ttonde{\boldsymbol{v},\boldsymbol{w}}$.\\
(Uniqueness under (sBF)$+$(sAB)) Suppose that $\boldsymbol{f}\ttonde{0}$ is bounded, and there exists a constant $M$ such that $\abs{\boldsymbol{{f}}\ttonde{\boldsymbol{z}}-\boldsymbol{{f}}\ttonde{\boldsymbol{z}'}}\leq M \lvert \boldsymbol{z}- \boldsymbol{z}' \lvert\ttonde{\lvert \boldsymbol{z}\lvert + \lvert \boldsymbol{z}'\lvert }$ for all $\boldsymbol{z},\boldsymbol{z}'\in\R^{N\times d}$.
Then the solution $\ttonde{\boldsymbol{v},\boldsymbol{w}}$ is unique in the class of bounded continuous solutions. 
\end{teo}
With the same notation of \textit{Remark \ref{formavecdelsystNagent}} we can write the system in (\ref{equilibriumNagent}) as
\begin{align*}
    \dfr\boldsymbol{Y}\ttonde{t} = \boldsymbol{z}\ttonde{t}\dfr \boldsymbol{W}\ttonde{t} + \boldsymbol{f}\ttonde{\boldsymbol{z}\ttonde{t}}\dfr t,\qquad \boldsymbol{Y}\ttonde{T} = \boldsymbol{e}_T.
\end{align*}
To verify condition (sAB) let us introduce an invertible linear transformation on $\R^N$ via 
\begin{align*}
    \Tilde{Y}^{\ttonde{i}} =  {Y}^{\ttonde{i,N}} + \Bar{Y}_N,\qquad i=1,\dots,N-1,\qquad \Tilde{Y}^{\ttonde{N}} = \Bar{Y}_N,
\end{align*}
define also 
\begin{align*}
    \boldsymbol{\tilde{z}} = \begin{pmatrix}
\boldsymbol{z}_{11} + \boldsymbol{\lambda}\ttonde{\boldsymbol{z}}& \boldsymbol{z}_{12} + \frac{1}{N} \boldsymbol{z}_{12}& \frac{1}{N} \boldsymbol{z}_{23} &\dots& \frac{1}{N} \boldsymbol{z}_{N N+1} \\
\boldsymbol{z}_{21} + \boldsymbol{\lambda}\ttonde{\boldsymbol{z}}& \frac{1}{N} \boldsymbol{z}_{12}&\boldsymbol{z}_{23} +\frac{1}{N} \boldsymbol{z}_{23} &\dots& \frac{1}{N} \boldsymbol{z}_{N N+1}\\
\vdots\\
\boldsymbol{\lambda}\ttonde{\boldsymbol{z}}& \frac{1}{N} \boldsymbol{z}_{12}& \frac{1}{N} \boldsymbol{z}_{23} &\dots &  \frac{1}{N} \boldsymbol{z}_{N N+1}
\end{pmatrix}
\end{align*}
and
\begin{align*}
   \ttonde{\boldsymbol{\tilde{e}}_T}^T:=\ttonde{\boldsymbol{e}_T^{\ttonde{1}} + \frac{1}{N}\sum_{i=1}^N\boldsymbol{e}_T^{\ttonde{i}},\boldsymbol{e}_T^{\ttonde{2}} + \frac{1}{N}\sum_{i=1}^N\boldsymbol{e}_T^{\ttonde{i}},\dots,\boldsymbol{e}_T^{\ttonde{N}} + \frac{1}{N}\sum_{i=1}^N\boldsymbol{e}_T^{\ttonde{i}} }
\end{align*}
A simple calculation reveals the dynamics of $\boldsymbol{\tilde{Y}}$ to be  
\begin{align*}
    \dfr\boldsymbol{\tilde{Y}}\ttonde{t} = \boldsymbol{\tilde{z}}\ttonde{t}\dfr \boldsymbol{W}\ttonde{t} + \boldsymbol{\tilde{f}}\ttonde{\boldsymbol{\tilde{z}}\ttonde{t}}\dfr t,\qquad \boldsymbol{\tilde{Y}}\ttonde{T} = \boldsymbol{\tilde{e}}_T,
\end{align*}
where 
\begin{align*}
    \boldsymbol{\tilde{f}}\ttonde{\boldsymbol{\tilde{z}}} = \begin{pmatrix}
    \boldsymbol{\tilde{z}}_{N1} \ttonde{\boldsymbol{\tilde{z}}_{11}-\boldsymbol{\tilde{z}}_{N1}}\\
    \boldsymbol{\tilde{z}}_{N1} \ttonde{\boldsymbol{\tilde{z}}_{21}-\boldsymbol{\tilde{z}}_{N1}}\\
    \vdots\\
    \frac{1}{2}\boldsymbol{\tilde{z}}_{N1}^2\\
\end{pmatrix}.
\end{align*}
Now let $\ttonde{\mathfrak{e}_1,\dots,\mathfrak{e}_N}$ be the standard Euclidean basis of $\R^N$, and $\boldsymbol{a}_{N+1} = \ttonde{\frac{1}{N},\dots,\frac{1}{N}}$. 
Then the set $\ttonde{-\mathfrak{e}_1,\dots,-\mathfrak{e}_N,\boldsymbol{a}_{N+1}}$ positively spans $\R^N$.
Moreover 
\begin{alignat*}{2}
    &-\mathfrak{e}_k^T\boldsymbol{\tilde{f}}\ttonde{\boldsymbol{z}}&&\leq \frac{1}{2}\abs{\boldsymbol{\tilde{z}}_{k1}}^2\leq \frac{1}{2} \abs{\mathfrak{e}_k^T\boldsymbol{\tilde{z}}}^2 \qquad \text{for }1\leq k\leq N\qquad \text{and}\\
    &\boldsymbol{a}_{N+1}\boldsymbol{\tilde{f}}\ttonde{\boldsymbol{z}} &&= \boldsymbol{\lambda}\ttonde{\boldsymbol{z}}\Big({\boldsymbol{\lambda}\ttonde{\boldsymbol{z}} - \frac{1}{N}\boldsymbol{z}_{N1}}\Big) + \frac{1}{2}N \boldsymbol{\lambda}\ttonde{\boldsymbol{z}}^2 \leq \frac{1}{2} \Big({\Big({\boldsymbol{\lambda}\ttonde{\boldsymbol{z}} - \frac{1}{N}\boldsymbol{z}_{N1}}\Big) + \boldsymbol{\lambda}\ttonde{\boldsymbol{z}}  }\Big)^2\\
    & &&\leq \frac{1}{2} \abs{\boldsymbol{a}_{N+1}\boldsymbol{\tilde{z}}}^2.
\end{alignat*}
Furthermore, condition (sAB) is invariant under invertible linear transformation of $\R^N$.
In order to verify condition (sBF) let us introduce an invertible linear transformation on $\R^N$ via 
\begin{align*}
    \bar{Y}^{\ttonde{i}} = Y^{\ttonde{i,N}} - Y^{\ttonde{N,N}},\qquad i=1,\dots,N-1,\qquad \bar{Y}^{\ttonde{N}} = Y^{\ttonde{N,N}},
\end{align*}
define also 
\begin{align*}
    \boldsymbol{\bar{z}} = \begin{pmatrix}
\boldsymbol{z}_{11} - \boldsymbol{z}_{N1} & \boldsymbol{z}_{12} & 0 &\dots&  &0& - \boldsymbol{z}_{NN+1} \\
\boldsymbol{z}_{21} -\boldsymbol{z}_{N1}  & 0 & \boldsymbol{z}_{23} & 0&\dots& 0& -\boldsymbol{z}_{NN+1}\\
\vdots\\
\boldsymbol{z}_{N1}& 0&\dots &  & & 0& \boldsymbol{z}_{NN+1}
\end{pmatrix},
\end{align*}
and $\boldsymbol{\bar{e}}_T$ is defined in a similar manner.
A simple calculation reveals the dynamics of $\boldsymbol{\bar{Y}}$ to be 
\begin{align*}
    \dfr\boldsymbol{\bar{Y}}\ttonde{t} = \boldsymbol{\bar{z}}\ttonde{t}\dfr \boldsymbol{W}\ttonde{t} +  \boldsymbol{\bar{f}}\ttonde{\boldsymbol{\bar{z}}\ttonde{t}}\dfr t,\qquad \boldsymbol{\bar{Y}}\ttonde{T} = \boldsymbol{\bar{e}}_T,
\end{align*}
where 
\begin{align*}
    \boldsymbol{\bar{f}}\ttonde{\boldsymbol{\bar{z}}} = \begin{pmatrix}
    \frac{1}{N}\big({\sum_{i=1}^{N-1}\ttonde{ \boldsymbol{\bar{z}}_{i1} + \boldsymbol{\bar{z}}_{N1}} + \boldsymbol{\bar{z}}_{N1}}\big) \boldsymbol{\bar{z}}_{11} \\
    \frac{1}{N}\big({\sum_{i=1}^{N-1}\ttonde{ \boldsymbol{\bar{z}}_{i1} + \boldsymbol{\bar{z}}_{N1}} + \boldsymbol{\bar{z}}_{N1}} \big)\boldsymbol{\bar{z}}_{21}\\
    \vdots\\
    -\frac{1}{2}N^2\big({\sum_{i=1}^{N-1}\ttonde{ \boldsymbol{\bar{z}}_{i1} + \boldsymbol{\bar{z}}_{N1}} + \boldsymbol{\bar{z}}_{N1}}\big)^2 + \frac{1}{N} \big({\sum_{i=1}^{N-1}\ttonde{ \boldsymbol{\bar{z}}_{i1} + \boldsymbol{\bar{z}}_{N1}} + \boldsymbol{\bar{z}}_{N1}}\big)\boldsymbol{\bar{z}}_{N1}\\
\end{pmatrix}.
\end{align*}
Now let $\gamma\ttonde{\boldsymbol{\bar{z}}}:=\frac{1}{N}\big({\sum_{i=1}^{N-1}\ttonde{ \boldsymbol{\bar{z}}_{i1} + \boldsymbol{\bar{z}}_{N1}} + \boldsymbol{\bar{z}}_{N1}}\big)$
\begin{align*}
    \boldsymbol{I}\ttonde{\boldsymbol{\bar{z}}}:=\begin{pmatrix}
    \gamma\ttonde{\boldsymbol{\bar{z}}} & \gamma\ttonde{\boldsymbol{\bar{z}}}& \dots&\gamma\ttonde{\boldsymbol{\bar{z}}}\\
    0& 0&0&0\\
    \vdots & \vdots& \vdots & \vdots\\
    0& 0&0&0
    \end{pmatrix}
\end{align*}
and observe that 
\begin{align*}
    \abs{\gamma\ttonde{\boldsymbol{\bar{z}}}}\leq \sum_{i=1}^N \abs{ \boldsymbol{\bar{z}}_{i1}}\leq C \abs{\boldsymbol{\bar{z}}}.
\end{align*}
Finally let 
\begin{align*}
    \ttonde{\boldsymbol{q}\ttonde{\boldsymbol{\bar{z}}}}^T:=\big({0,\dots,0,-\frac{1}{2}\gamma\ttonde{\boldsymbol{\bar{z}}}^2}\big)
\end{align*}
and observe that 
\begin{align*}
    \abs{\gamma\ttonde{\boldsymbol{\bar{z}}}}^2\leq C \sum_{i=1}^N\abs{\boldsymbol{\bar{z}}_{i1}}^2 \leq C \sum_{i=1}^N\lvert{\boldsymbol{\bar{z}^i}}\lvert^2.
\end{align*}
Using this, one easily checks that $\boldsymbol{\bar{f}}$ satisfies the condition (sBF).
Lastly we show that there exists $M>0$ such that 
\begin{align*}
    \abs{\boldsymbol{\bar{f}}\ttonde{\boldsymbol{z}}-\boldsymbol{\bar{f}}\ttonde{\boldsymbol{z}'}}\leq M \lvert \boldsymbol{z}- \boldsymbol{z}' \lvert\ttonde{\lvert \boldsymbol{z}\lvert + \lvert \boldsymbol{z}'\lvert }
\end{align*}
take $\boldsymbol{z},\boldsymbol{z'}\in\R^{N\times N+1}$ and recall that
\begin{align*}
    \boldsymbol{\bar{f}}\ttonde{\boldsymbol{z}}=\begin{pmatrix}
    \gamma\ttonde{\boldsymbol{z}} \boldsymbol{{z}}_{11} \\
    \gamma\ttonde{\boldsymbol{z}} \boldsymbol{{z}}_{21}\\
    \vdots\\
    -\frac{1}{2}\gamma\ttonde{\boldsymbol{z}}^2 + \gamma\ttonde{\boldsymbol{z}}\boldsymbol{{z}}_{N1}\\
\end{pmatrix}, \boldsymbol{\bar{f}}\ttonde{\boldsymbol{z}'}=\begin{pmatrix}
    \gamma\ttonde{\boldsymbol{z}'} \boldsymbol{{z}'}_{11} \\
    \gamma\ttonde{\boldsymbol{z}'} \boldsymbol{{z}'}_{21}\\
    \vdots\\
    -\frac{1}{2}\gamma\ttonde{\boldsymbol{z}'}^2 + \gamma\ttonde{\boldsymbol{z}'}\boldsymbol{{z}'}_{N1}\\
\end{pmatrix}.
\end{align*}
In order to estimate $\abs{\boldsymbol{\bar{f}}\ttonde{\boldsymbol{z}}-\boldsymbol{\bar{f}}\ttonde{\boldsymbol{z}'}}$ we notice that 
\begin{align*}
    \abs{\gamma\ttonde{\boldsymbol{z}} - \gamma\ttonde{\boldsymbol{z}'}} = \frac{1}{N} \abs{\sum_{i=1}^{N-1}\ttonde{\ttonde{\boldsymbol{z}_{i1} - \boldsymbol{z}'_{i1}} + \ttonde{\boldsymbol{z}_{N1}-\boldsymbol{z}'_{N1}}}+\ttonde{\boldsymbol{z}_{N1}-\boldsymbol{z}'_{N1}} }&\leq C\sum_{i=1}^N \lvert \boldsymbol{z}_{i1}- \boldsymbol{z}'_{i1}\lvert\\ &\leq C \lvert \boldsymbol{z}- \boldsymbol{z}' \lvert.
\end{align*}
and that we just need to estimate the last component which we do as follows
\begin{align*}
    &\frac{1}{2}\abs{\gamma\ttonde{\boldsymbol{z}}^2-\gamma\ttonde{\boldsymbol{z}'}^2} + \abs{\gamma\ttonde{\boldsymbol{z}}\boldsymbol{{z}}_{N1}- \gamma\ttonde{\boldsymbol{z}'}\boldsymbol{{z}'}_{N1}}\\
    &\leq C\abs{\gamma\ttonde{\boldsymbol{z}} - \gamma\ttonde{\boldsymbol{z}'}}  \abs{\gamma\ttonde{\boldsymbol{z}} + \gamma\ttonde{\boldsymbol{z}'}} + \abs{\gamma\ttonde{\boldsymbol{z}}\boldsymbol{{z}}_{N1} - \gamma\ttonde{\boldsymbol{z}'}\boldsymbol{{z}}_{N1}} + \abs{\gamma\ttonde{\boldsymbol{z}'}\boldsymbol{{z}}_{N1} - \gamma\ttonde{\boldsymbol{z}'}\boldsymbol{{z}'}_{N1}}\\
    &\leq C_2\lvert \boldsymbol{z}- \boldsymbol{z}' \lvert\ttonde{\lvert \boldsymbol{z}\lvert + \lvert \boldsymbol{z}'\lvert } + \lvert \boldsymbol{{z}}_{N1} \lvert \abs{\gamma\ttonde{\boldsymbol{z}} - \gamma\ttonde{\boldsymbol{z}'}} + \abs{\gamma\ttonde{\boldsymbol{z}'}}\abs{\boldsymbol{{z}}_{N1} - \boldsymbol{{z}'}_{N1}}\leq C_3 \lvert \boldsymbol{z}- \boldsymbol{z}' \lvert\ttonde{\lvert \boldsymbol{z}\lvert + \lvert \boldsymbol{z}'\lvert }.
\end{align*}
Applying {Theorem \ref{exandunsolsyst}} we get our result.
\end{proof}
\printbibliography
\end{document}